\documentclass[reprint,twocolumn,aps,amsmath,amssymb,pra,superscriptaddress,floatfix,tightenlines,nofootinbib,nobibnotes]{revtex4-2}

\usepackage{graphicx}
\usepackage[colorlinks,linkcolor=blue,urlcolor=blue, citecolor=blue]{hyperref}
\usepackage{xurl}
\usepackage{booktabs}
\usepackage{subcaption}
\usepackage{comment}
\usepackage{nicematrix}
\usepackage{tikz}
\usetikzlibrary{positioning, shapes.multipart}

\usepackage{dsfont}
\usepackage{physics}
\usepackage{nicefrac}
\usepackage{mathtools}
\usepackage{amsmath, amssymb, amsfonts, amsthm}
\usepackage{thmtools, thm-restate}

\usepackage{algorithmic}
\usepackage{algorithm}

\newtheorem{theorem}{Theorem}
\newtheorem*{theorem*}{Theorem}

\newcommand{\argmin}{\operatorname*{argmin}}

\DeclareMathOperator{\polylog}{polylog}
\DeclareMathOperator{\poly}{poly}

\allowdisplaybreaks

\begin{document}

\title{Overflow-Safe Polylog-Time Parallel Minimum-Weight Perfect Matching Decoder: Toward Experimental Demonstration}

\author{Ryo Mikami}
\email{ryo-223606@g.ecc.u-tokyo.ac.jp}
\affiliation{
Department of Computer Science, Graduate School of Information Science and Technology, The University of Tokyo, 7--3--1 Hongo, Bunkyo-ku, Tokyo, 113--8656, Japan
}
\author{Hayata Yamasaki}
\email{hayata.yamasaki@gmail.com}
\affiliation{
Department of Computer Science, Graduate School of Information Science and Technology, The University of Tokyo, 7--3--1 Hongo, Bunkyo-ku, Tokyo, 113--8656, Japan
}

\begin{abstract}
    Fault-tolerant quantum computation (FTQC) requires fast and accurate decoding of quantum errors, which is often formulated as a minimum-weight perfect matching (MWPM) problem. A determinant-based approach has been proposed as a novel method to surpass the conventional polynomial runtime of MWPM decoding via the blossom algorithm, asymptotically achieving polylogarithmic parallel runtime. However, the existing approach requires an impractically large bit length to represent intermediate values during the computation of the matrix determinant; moreover, when implemented on a finite-bit machine, the algorithm cannot detect overflow, and therefore, the mathematical correctness of such algorithms cannot be guaranteed. In this work, we address these issues by presenting a polylog-time MWPM decoder that detects overflow in finite-bit representations by employing an algebraic framework over a truncated polynomial ring. Within this framework, all arithmetic operations are implemented using bitwise XOR and shift operations, enabling efficient and hardware-friendly implementation. Furthermore, with algorithmic optimizations tailored to the structure of the determinant-based approach, we reduce the arithmetic bit length required to represent intermediate values in the determinant computation by more than $99.9\%$, making it possible to implement it on machines supporting $512$-bit computing while preserving the decoder's polylog runtime scaling. These results open the possibility of a proof-of-principle demonstration of the polylog-time MWPM decoding in the early FTQC regime.
\end{abstract}

\maketitle

\tableofcontents

\section{Introduction}
\label{Sec:Introduction}
\subsection{Background and motivation}

Practical validation of theoretically established methods with finite-size devices, even when demonstrated in a small testbed, has played an important role in the development of quantum error correction (QEC).
Fault-tolerant quantum computation (FTQC) provides a framework for reliable quantum computation in the presence of noise, and among quantum error-correcting codes, the surface code is regarded as one of the leading approaches~\cite{fowler2012surface} due to its high error tolerance and practical implementability.
In theory, the threshold theorem guarantees that the logical error rate can be suppressed exponentially as we increase the distance $d$ of the surface code~\cite{dennis2002topological,fowler2012proof}.
More recently, experimental demonstrations of this error-suppression scaling in surface code at small code distances of $d = 3, 5, 7$ have had a significant impact~\cite{google2023suppressing,google2025quantum}.
Such small-scale validations serve as important milestones toward scalable FTQC\@. 

In practice, such demonstrations require a solid theoretical foundation for decoding.
The surface code preserves logical quantum states by extracting syndromes that indicate where errors have occurred and correcting those errors accordingly~\cite{dennis2002topological}.
QEC is performed by running classical decoding algorithms.
The time overhead of the classical decoder has a non-negligible impact on the overall runtime of FTQC~\cite{RevModPhys.87.307,takada2025doubly}, while its decoding performance directly affects the logical error rate of the computation.
In general, degenerate quantum maximum likelihood decoding (DQMLD), while achieving the optimal error-correction performance, is computationally hard~\cite{PhysRevA.83.052331,iyer2015hardness}.
In contrast, for the surface code, a decoding method based on minimum-weight perfect matching (MWPM) runs in polynomial time and provides a theoretical guarantee on the error suppression~\cite{dennis2002topological,fowler2012proof}.
There are other decoders with similar theoretical guarantees, such as the union-find (UF) decoder~\cite{delfosse2021almost,yoshida2026prooffinitethresholdunionfind}, the cellular-automaton (CA) decoder~\cite{kubica2019cellular,harrington2004analysis,herold2015cellular,vasmer2021cellular,balasubramanian2024local,lake2025fast,lake2025local}, and the renormalization-group (RG) decoder~\cite{bravyi2011analytic}, but their error suppression is weaker than MWPM decoding.
Heuristic decoders, such as belief propagation with ordered statistic decoding (BPOSD)~\cite{Panteleev2021degeneratequantum, PhysRevResearch.2.043423} and neural-network decoders~\cite{varsamopoulos2017decoding}, may achieve strong error suppression but do not provide the theoretical guarantee.
As a result, MWPM decoding for the surface code plays an important role in validating and demonstrating the feasibility of FTQC with the surface code\@.

A well-known limitation of conventional MWPM decoding is its large polynomial runtime exponent, which constrains decoding throughput and can become a bottleneck for the time overhead of FTQC~\cite{takada2025doubly}.
Conventionally, the decoder solves the MWPM problem using Edmonds' blossom algorithm~\cite{edmonds1965paths,edmonds1965maximum}.
The blossom algorithm proposed by Edmonds has a computational complexity of $O(|V|^4)$, where $|V|$ denotes the number of vertices in the graph. Micali and Vazirani~\cite{micali1980v,vazirani2020proofmvmatchingalgorithm} achieved a complexity of $O(\sqrt{|V|}\cdot |E|)$ by introducing techniques to accelerate the search for augmenting paths, where $|E|$ denotes the number of edges.
Kolmogorov~\cite{kolmogorov2009blossom} developed a faster implementation that is effective in practice by incorporating dual variables and priority queues.
In the context of surface-code decoding, representative implementations of MWPM decoders, such as sparse blossom~\cite{Higgott2025sparseblossom} and fusion blossom~\cite{wu2023fusion}, are also based on variants of the blossom algorithm.
Reference~\cite{fowler2013minimum} considered the problem of solving the MWPM decoding for the surface code within a short average parallel runtime, but in the worst case, it still requires a polynomial time complexity in $|V|$.
None of these blossom-algorithm-based methods can achieve polylog time complexity; even with parallelization, no blossom-algorithm-based method is currently known to achieve polylog parallel running time. 

More recently, Ref.~\cite{takada2025doubly} proposed an alternative algorithm for the MWPM decoding that asymptotically achieves polylog parallel runtime, substantially improving the conventional polynomial scaling of the runtime of MWPM decoders based on the blossom algorithms.
The MWPM decoder proposed in Ref.~\cite{takada2025doubly} is based on the algorithm of Mulmuley et al.~\cite{mulmuley1987matching,Mulmuley1987}, computing determinants of matrices obtained from the detector graph, where the matrix entries are integers that may take values exponential in the detector graph's edge weights (i.e., polynomial in the bit length).

The practical implementation of this polylog-time MWPM decoding, however, faces two fundamental challenges.
First, the determinant-based formulation assumes integer arithmetic that may require an impractically large number of bits to represent intermediate values during the computation of the matrix determinant.
Second, under finite-precision hardware constraints, overflow may silently corrupt intermediate values, thereby breaking the algorithm's theoretical correctness guarantee.
Despite the challenges, in the same spirit as the proof-of-principle demonstrations of error suppression in QEC~\cite{google2023suppressing, google2025quantum}, demonstrating the practical feasibility of this asymptotically nontrivial runtime scaling of MWPM decoding, even at small code distances, is valuable for the continued advancement of QEC\@.

\begin{figure}[tbp]
  \centering
  \includegraphics[width=3.4in]{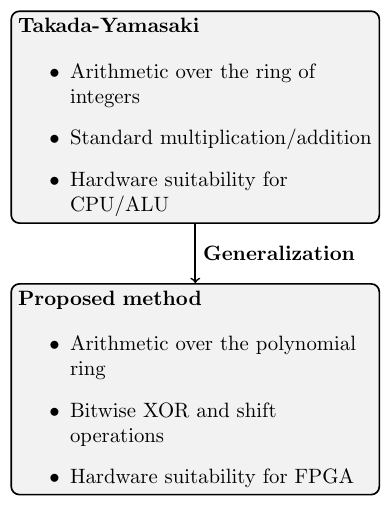}
  \caption{Conceptual comparison between Takada-Yamasaki~\cite{takada2025doubly} and the proposed bitwise algorithm.}
  \label{fig:comparison}
\end{figure}

\subsection{Summary of main results}

In this work, we address these challenges in the existing polylog-time parallel MWPM-decoding method by constructing and analyzing an approach that is feasible to implement at the same scale as current surface-code experiments.
Our method generalizes the previous determinant-based framework from integer arithmetic to bitwise polynomial arithmetic, which is fully implementable using bitwise operations while asymptotically preserving polylog time complexity.
To clarify the distinction between our method and the previous approach in Ref.~\cite{takada2025doubly}, we summarize the key contrasts in Fig.~\ref{fig:comparison}.

Our contributions are summarized as follows.

\begin{enumerate}
\item
We develop an algebraic framework that analytically guarantees the correctness of our determinant-based decoding method under a given finite-bit-length constraint. We reformulate the determinant computation over a truncated polynomial ring, where finite-bit arithmetic corresponds to polynomial truncation. In this framework, overflow becomes mathematically well-defined and naturally detectable during the decoding method. We prove that whenever the MWPM weight lies within the representable range, our decoding algorithm outputs exactly the same result as the previous MWPM-decoding method~\cite{takada2025doubly}; otherwise, i.e., if intermediate values during the determinant computation overflow on a fixed finite-bit machine, it explicitly signals a logical error. All arithmetic operations reduce to bitwise XOR and shift operations, making our approach naturally suitable for FPGA-style parallel architectures while asymptotically preserving polylog parallel depth, as theoretically proposed in Ref.~\cite{takada2025doubly}.
\item
Second, we remove dominant sources of bit-length blowup inherent in the previous procedure in Ref.~\cite{takada2025doubly}. First, we eliminate a costly weight-amplification procedure, a key modification from the previous procedure. Building upon this, we further introduce a variable-precision scheme that first performs a costly part, i.e., computing candidates of MWPMs, with low precision, and later efficiently verifies them with high precision, where the verification is at a low cost. Together, these refinements substantially reduce the required arithmetic bit length while maintaining the desired suppression of the logical error rate. 
In a proof-of-principle numerical simulation on surface codes with code distance $d=5$, we reduce the required arithmetic bit length for polylog-time MWPM decoding from $\lessapprox 6\times 10^5$ bits in the previous method~\cite{takada2025doubly} to $\lessapprox 5\times 10^2$ bits, i.e., by up to $99.9\%$, while keeping the decoding failure probability below the logical error rate for $d=5$.
The required number of classical processors remains sub-linear in the size of the detector graph.
\end{enumerate}

Our results establish the overflow-safe, hardware-feasible realization of asymptotically polylog-time parallel MWPM decoding, narrowing the gap between asymptotic optimality and practical implementation in FTQC architectures. Combined with prior results showing that polylog-time parallel decoding implies doubly-polylog-time-overhead FTQC~\cite{takada2025doubly}, our work strengthens the practical feasibility of such highly time-efficient FTQC, opening up the possibility of its proof-of-principle demonstration in the early FTQC regime.

\subsection{Organization of the paper}
The remainder of this paper is organized as follows.
In Sec.~\ref{Sec:Theoretically_guaranteeed} we construct our determinant-based polylog-time MWPM decoding algorithm, analyze its correctness, and present numerical results to evaluate its parameter choice.
In Sec.~\ref{Sec:Heuristic_modification}, we further optimize this algorithm by introducing heuristic techniques for substantially reducing the required arithmetic bit length, where we do not analytically provide the conditions on securing correctness, while also presenting the results of numerical simulations to support its correctness and propose appropriate parameter choices.
Finally, in Sec.~\ref{Sec:Discussion}, we discuss the implications of our results for low-overhead FTQC and outlines directions for future research.

\section{Overflow-safe polylog-time MWPM decoder}
\label{Sec:Theoretically_guaranteeed}

In this section, we first describe the proposed decoding algorithm based on polynomial arithmetic over a truncated polynomial ring in Sec.~\ref{Sec:CorrectAlgorithm}. In Sec.~\ref{Sec:Analysis}, we analyze its theoretical properties, including correctness and failure detectability under finite-bit representations. In Sec.~\ref{Sec:EvaluationTheoretically}, we provide numerical simulations to estimate the bit precision for reliable decoding.

\subsection{Description of the polylog-time parallel algorithm for MWPM decoding}\label{Sec:CorrectAlgorithm}
We describe the details of our polylog-time parallel algorithm, building on the approach that was introduced in Ref.~\cite{takada2025doubly} (see also Appendix~\ref{appendix:A} for the details of the approach in Ref.~\cite{takada2025doubly}). In Sec.~\ref{Sec:Task_of_decoding}, we describe the decoding task for the surface code. In Sec.~\ref{Sec:Polynomial_rings}, we introduce a mathematical framework that serves as the foundation of our methods. In Sec.~\ref{Sec:Parallel_algorithm}, we propose our polylog-time parallel algorithm for MWPM decoding.

\subsubsection{Task of decoding}\label{Sec:Task_of_decoding}
We summarize the decoding task for the surface code, which can be, in the context of FTQC, formulated as an MWPM problem on a weighted graph (see also, e.g., Ref.~\cite{takada2025doubly} for the details of the problem formulation). The surface code performs error correction by repeatedly applying syndrome extraction, which is represented as a syndrome-extraction circuit composed of state preparations, gates, measurements, and idle operations on physical qubits. We refer to each operation in this circuit as a location. Following the convention of Ref.~\cite{gottesman2013fault}, we assume a local stochastic Pauli error model. Namely, each location $j$ has an error parameter $p_j$, and for any set $A$ of locations, the probability that all locations in $A$ are faulty is upper bounded as
\begin{align}
    \Pr[\text{all locations in } A \text{ are faulty}] \le \prod_{j \in A} p_j. \label{Eq:WeightUpperBound}
\end{align}
Pauli errors may occur at faulty locations.

A syndrome measurement consists of measuring stabilizer generators, each of which is represented as a tensor product of Pauli operators, yielding a classical single-bit outcome, i.e., $0$ or $1$.
The syndrome-extraction circuit provides detectors, which are the parities of pairs of consecutive syndrome measurement outcomes used for error correction.
The detector is $0$ if there is no error, and physical errors may flip a set of detectors into $1$, which are called active detectors.
Using the detectors as vertices, we form a detector graph in which error chains are represented as paths over detectors.
However, error chains may also terminate on spatial or time-like boundaries of the decoding window in the syndrome-extraction circuit.
Such boundaries are addressed by introducing boundary vertices in the detector graph, which can serve as endpoints of error chains.
In the surface code, an error at a given location in the syndrome-extraction circuit may flip either one or two detectors; these cases are represented by edges connecting a detector to a boundary vertex, or connecting a pair of detectors, respectively.

Using these detectors and boundary vertices, we define a detector graph~\cite{Higgott2025sparseblossom}
\begin{align}
    G \coloneqq (V, E),
\end{align}
where each vertex represents either a detector or a boundary vertex. An edge $e\in E$ connects either one detector and one boundary, or two detectors. Edge weights $w(e)$ are defined by 
\begin{align}
    w(e) \coloneqq \left\lceil -C\ln\left(\sum_j p_j\right)\right\rceil, \label{Eq:discretized}
\end{align}
where $C\geq1$ is a scaling factor used for integer rounding, $\lceil \cdot\rceil$ denotes the ceiling function, $p_j$ denotes the error parameter of a location $j$ in~\eqref{Eq:WeightUpperBound}, and the sum is taken over multiple possible locations that may flip the detectors incident with the edge $e$.
This form corresponds to the negative log-likelihood of the error and ensures that the MWPM problem is equivalent to finding the most probable set of error chains.

From a detector graph, we construct a path graph (see also, e.g., Ref.~\cite{takada2025doubly} for the formal description of this procedure)
\begin{align}
    \overline{G} = \qty(\overline{V}, \overline{E}).
\end{align}
The weights of the path graph $w : \overline{E} \to \mathbb{Z}$ are obtained from the weights of the detector graph using Dijkstra's algorithm~\cite{dijkstra2022note}. Each edge in the path graph corresponds to a possible error chain between two detection events or between a detection event and a boundary vertex.

Given a path graph $\overline{G} = \qty(\overline{V}, \overline{E})$ with integer weights $w$, the decoding task is to find a perfect matching $M^* \subseteq \overline{E}$ that minimizes the total weight
\begin{align}
    M^* = \argmin_{M \subseteq \overline{E}} \sum_{e \in M} w(e).
\end{align}
The MWPM $M^*$ corresponds to the most probable configuration of physical errors consistent with the observed detection events and thus determines the appropriate recovery operation on the code block.

To summarize, the inputs to the decoding algorithm are a path graph $\overline{G} = \qty(\overline{V}, \overline{E})$ and a weight function $w : \overline{E} \to \mathbb{Z}$, and the output is an MWPM $M^*$. Finding such a matching feasibly in polylog parallel time is the goal of decoding in this work.

\subsubsection{Finite-digit binary numbers and truncated polynomial rings}\label{Sec:Polynomial_rings}

In a modern computer, numbers are represented as finite-digit binary values. Due to hardware constraints, digits beyond a certain length are truncated as overflow. Such behavior does not occur in the ring of integers; therefore, the assumption that overflow never occurs underlies the guarantee of mathematical correctness in the previous analysis~\cite{takada2025doubly}, which may indeed be violated in practical implementations due to finite arithmetic bit lengths.

To preserve mathematical consistency under such practical constraints, we introduce the following natural correspondence
\begin{align}
    &{a_{n-1} a_{n-2} \cdots a_0}_{(2)} \nonumber \\
    &\xrightarrow{} a_{n-1}X^{n-1} + a_{n-2}X^{n-2} + \cdots + a_0.
\end{align}
Here, the left-hand side represents the binary representation of an integer, where $a_i \in \{0, 1\}$ and the representable bit length is $n$. The fact that digits beyond $n$ bits are discarded due to overflow corresponds to the condition that $X^n = 0$ holds, or equivalently, to taking the remainder upon the division of a polynomial by $X^n$. Therefore, a finite-digit binary number can be rigorously represented as the remainder of a polynomial over the finite field $\mathbb{F}_2 = \{0, 1\}$ modulo $X^n$. The set of all such polynomials forms a ring denoted by $\mathbb{F}_2[X]/(X^n)$.

Using this algebraic structure, arithmetic operations such as addition and multiplication over the integers can be consistently defined as addition and multiplication in $\mathbb{F}_2[X]/(X^n)$, respectively, incorporating overflow in an algebraically coherent manner. Although arithmetic in the ring of integers and that in the polynomial ring generally differ, as shown later in Theorem~\ref{thm:correctness}, our algorithm produces identical results under this correspondence. Consequently, computations with finite-bit representations on actual hardware can be modeled in a mathematically rigorous manner, thereby ensuring the overall correctness of our algorithm.

Although our algorithm will be described using polynomials, it can be implemented using bitwise operations, as shown by the following theorem.

\begin{theorem}[\label{thm:bitlevel}Bit-Level Computability]
    For any constant positive integers $w_\mathrm{th}\in\{1,2,\ldots\}$, polynomial arithmetic over $\mathbb{F}_2[X]/(X^{w_{\mathrm{th}}})$ can be implemented by bitwise XOR and shift operations. 
\end{theorem}

\begin{proof}
We first show that any element in $\mathbb{F}_2[X]/(X^{w_{\mathrm{th}}})$ can be represented as a bit string.
Since the coefficients in this polynomial ring take values in $\{0, 1\}$, and since $X^w = 0$ for any $w \ge w_{\mathrm{th}}$, each polynomial has finitely many terms. 
Therefore, by directly mapping each coefficient to its corresponding bit, with $1$ indicating the presence and $0$ the absence of a term, every element can be expressed as a bit string of length $w_{\mathrm{th}}$. 
We then show that addition and multiplication operations can be implemented using bitwise operations.

\smallskip
\textit{Addition.}
For $a, b \in \mathbb{F}_2[X]/(X^{w_{\mathrm{th}}})$ and for each $0 \le n < w_{\mathrm{th}}$,
\begin{align}
    (a + b)[X^n] =
    \begin{cases}
        1 & (a[X^n] \ne b[X^n]), \\
        0 & (a[X^n] = b[X^n]),
    \end{cases}
\end{align}
where $p[X^n]$ denotes the coefficient of $X^n$ in $p$.
This relation is equivalent to
\begin{align}
    a + b = a_{\text{bit}} \oplus b_{\text{bit}},
\end{align}
where $p_{\text{bit}}$ denotes the corresponding bitwise representation of $p$, and $\oplus$ denotes the bitwise XOR.
Thus, the addition in $\mathbb{F}_2[X]/(X^{w_{\mathrm{th}}})$ corresponds exactly to a bitwise XOR operation.

\smallskip
\textit{Multiplication.}
For $a, b \in \mathbb{F}_2[X]/(X^{w_{\mathrm{th}}})$, we have
\begin{align}
    a \times b = \sum_{i=0}^{w_{\mathrm{th}}-1} \big( a \cdot b[X^i] X^i \big).
\end{align}
This can be rewritten as
\begin{align}
    a \times b = \bigoplus_{i:\, b[X^i] = 1} (a_{\text{bit}} \ll i),
\end{align}
where $\ll$ denotes the left-shift operation.
Since whether $b[X^i] = 1$ can be determined directly from $b_{\text{bit}}$, multiplication is realized by a sequence of bitwise shifts and XORs.

Thus, all operations over $\mathbb{F}_2[X]/(X^{w_{\mathrm{th}}})$ can be implemented solely using bit-level XOR and shift operations, demonstrating that computations in this ring are implementable directly on bits.
\end{proof}

\subsubsection{A parallel algorithm for finding an MWPM via polynomial arithmetic}\label{Sec:Parallel_algorithm}

Given a path graph $\overline{G} = \qty(\overline{V}, \overline{E})$ with edge weights $w: \overline{E} \to \mathbb{Z}$, we propose a parallel algorithm for finding an MWPM in $\overline{G}$ based on polynomial arithmetic over a truncated polynomial ring $\mathbb{F}_2[X]/(X^{w_{\mathrm{th}}})$ with a fixed constant $w_\mathrm{th}$.
In this framework, the truncation threshold $w_{\mathrm{th}}$ is set in advance, prior to computation, based on the available computational resources and hardware constraints.

Our algorithm is described as follows.

\begin{enumerate}
  \item Adjust each edge weight and replace $w$ by $\tilde{w}$ for each $e \in \overline{E}$ as
  \begin{align}
      \tilde{w}(e) \coloneqq \tilde{C}w(e) + W(e), \label{Eq:tildew}
  \end{align}
  where we take the factor $\tilde{C}$ and the perturbation function $W$ as
  \begin{align}
      &\tilde{C} \coloneqq \frac{|\overline{V}|}{2}(W_{\max} - 1) + 1, \label{Eq:tildeC} \\
      &W: \overline{E} \xrightarrow{} \mathbb{N}, \\
      &1\leq W(e) \leq W_{\max},  \label{Eq:rangeofW} \\
      &W_{\max} \coloneqq \max_{e \in \overline{E}}\{W(e)\}. \label{Eq:Wmax} 
  \end{align}
  Here, for each $e \in \overline{E}$, the perturbation $W(e)$ is drawn independently and uniformly at random from the integer set $\{1, \ldots, W_{\max}\}$.
  With this perturbation, with high probability, one of the MWPMs in $\overline{G}$ becomes a unique MWPM in the resulting weighted graph\@.
As will be shown later in Theorem~\ref{thm:correctness}, if the resulting weighted graph has a unique MWPM, then the algorithm's output subset $M$ represents this unique MWPM, which is one of the MWPMs of $\overline{G}$, similar to the previous algorithm in Ref.~\cite{takada2025doubly} based on integer arithmetic.
  \item Let $A$ be a Tutte matrix with its element $A_{i, j}$ given by 
  \begin{align}
      A_{i, j} = 
      \begin{cases}
          x_{i, j} & (i < j, \{i, j\} \in \overline{E}), \\
          -x_{i, j} & (i > j, \{i, j\} \in \overline{E}), \\
          0 & \text{otherwise},
      \end{cases}
  \end{align}
  where $x_{i, j}$ is a variable. Then, for each edge $\{i, j\} \in \overline{E}$, we assign 
  \begin{align}
      x_{i, j} = X^{\tilde{w}(\{i, j\})}, \label{Eq:indeterminate}
  \end{align}
  where $X$ is the indeterminate of $\mathbb{F}_2[X]/(X^{w_{\mathrm{th}}})$. Since we are working on $\mathbb{F}_2$ when computing the coefficient of $X^w$, we then obtain a matrix $B$ defined as
    \begin{align}
      B_{i, j} = 
      \begin{cases}
          X^{\tilde{w}(\{i, j\})} &\{i, j\} \in \overline{E}, \\
          0 & \text{otherwise}, \\
      \end{cases}
      \label{Eq:Tuttemat2}
  \end{align}
  from an assignment~\eqref{Eq:indeterminate}.
  \item Compute the determinant $\det(B)$ of $B$ and obtain
  \begin{align}
      w^* \coloneqq \frac{\deg_{\min}\left(\det(B)\right)}{2}, \label{Eq:minweight2}
  \end{align}
  where $\deg_{\min}(f)$ denotes the minimum exponent of $X$ appearing in the polynomial $f$.
  Note that the determinant of ring-valued matrices can be computed in polylog parallel runtime by the Samuelson-Berkowitz algorithm~\cite{berkowitz1984computing} (see also Sec.~S2 of Ref.~\cite{takada2025doubly} for the details of this method). 
  \item For each edge, calculate a minor $M^B_{i, j} = \det(B^{(i, j)})$, where $B^{(i, j)}$ is defined as an $\qty(|\overline{V}| - 1) \times \qty(|\overline{V}| - 1)$ submatrix obtained from $B$ by removing the $i$th row and the $j$th column. Finally, we obtain
  \begin{align}
      M = \left\{\{i, j\}\middle| \deg_{\min}\left(M^B_{i, j} \cdot X^{\tilde{w}(\{i, j\})}\right) = w^*\right\} \subseteq \overline{E}. \label{Eq:setofMWPM2}
  \end{align}
  When step~3 computes $w^* = 0$, the algorithm terminates and outputs a failure indicator, which is interpreted as a logical error.
\end{enumerate}

The following theorem shows that this algorithm can be parallelized to achieve a polylog parallel runtime.

\begin{theorem}[\label{thm:parallel}Parallelizability]
Given a path graph of size $\overline{V}$ and a perturbation range $W_{\max}$, all computational steps of the algorithm are parallelizable with polylog depth $O\qty(\polylog\qty(|\overline{V}|, |W_{\max}|))$ using bitwise operations.
\end{theorem}

\begin{proof}
We show that the proposed algorithm can be parallelized in polylog time.
First, in Step~4, each minor is computed independently and can therefore be evaluated in parallel using $O\qty(|\overline{V}|^2)$ processors.  
Thus, it suffices to show that the determinant of an $O\qty(|\overline{V}|) \times O\qty(|\overline{V}|)$ matrix can be computed in parallel within polylog time.  

As indicated above, it is known that the determinant of such a matrix with entries in polynomial rings can be computed in $O\qty(\log^2 |\overline{V}|)$ parallel time using $O(\mathrm{poly}(|\overline{V}|))$ processors, as achieved by the Samuelson--Berkowitz algorithm~\cite{berkowitz1984computing}.  
Since this algorithm does not require division, it can also be applied to matrices $B$ in~\eqref{Eq:Tuttemat2} whose entries are polynomials over $\mathbb{F}_2[X]/(X^{w_{\mathrm{th}}})$.  

From Theorem~\ref{thm:bitlevel}, each polynomial in $\mathbb{F}_2[X]/(X^{w_{\mathrm{th}}})$ can be represented as a bit string, and both addition and multiplication over this ring can be implemented using bitwise operations as the standard arithmetic operations. As the addition and multiplication of $N$-bit integers can be performed in $O(\log N)$ parallel time using $O(\mathrm{poly}(N))$ processors, the same parallel complexity applies to operations in $\mathbb{F}_2[X]/(X^{w_{\mathrm{th}}})$. In our case, from \eqref{Eq:tildew}, we have $N = O(\mathrm{poly}(|\overline{V}|, W_{\max}))$.  

Therefore, the total runtime for finding the unique MWPM is
\begin{align}
    O\qty(\polylog\qty(|\overline{V}|, |W_{\max}|)),
\end{align}
and the required number of processors is
\begin{align}
    O\qty(\poly\qty(|\overline{V}|, |W_{\max}|)).
\end{align}
\end{proof}

\subsection{Theoretical analysis of the proposed method}
\label{Sec:Analysis}
We analyze the theoretical properties of the proposed method, focusing on its correctness in Theorem~\ref{thm:correctness} and failure detectability in Theorem~\ref{thm:failure}.

First, the following theorem shows that our algorithm produces results identical to those of the existing MWPM decoder in Ref.~\cite{takada2025doubly} if no overflow occurs.

\begin{theorem}[\label{thm:correctness}Correctness]
    Let $w^*$ denote the MWPM weight for the input graph $\overline{G} = (\overline{V}, \overline{E})$ with edge weights $\tilde{w}: \overline{E} \to \mathbb{Z}$. If
    \begin{align}
        w^{*}<\frac{w_{\mathrm{th}}}{2}, 
    \end{align}
    then our algorithm outputs a unique and correct MWPM with high probability at least $1 - |V|/W_{\mathrm{max}}$.
\end{theorem}

\begin{proof}
We will check the correctness of each algorithmic step.

\smallskip
\textit{Step~1.}
In this step, the edge weights are adjusted. Let $M^*$ denote the MWPM of $\overline{G} = \qty(\overline{V}, \overline{E})$ with weights $\tilde{w}$, and let $M'$ be any perfect matching that is not an MWPM of $\overline{G}$ with weights $w$. 
Noting that $|M^*| = |M'| = |\overline{V}|/2$ and from~\eqref{Eq:rangeofW}, we have
\begin{align}
    \sum_{e \in M^*} \tilde{w}(e) &\le \tilde{C}\sum_{e \in M^*} w(e) + \frac{|\overline{V}|}{2} W_{\max}, \\
    \sum_{e \in M'} \tilde{w}(e) &\ge \tilde{C}\sum_{e \in M'} w(e) + \frac{|\overline{V}|}{2}.
\end{align}
Thus, for $M'$ and $M^*$, it holds that
\begin{align}
    &\sum_{e \in M'} \tilde{w}(e) - \sum_{e \in M^*} \tilde{w}(e) \nonumber \\
    &\ge \tilde{C}\left(\sum_{e \in M'} w(e) - \sum_{e \in M^*} w(e)\right) - \frac{|\overline{V}|}{2}(W_{\max} - 1) \\
    &\ge \tilde{C} - \frac{|\overline{V}|}{2}(W_{\max} - 1) \\
    &\ge 1.
\end{align}
This inequality shows that $M^*$ has a smaller weight than any other perfect matching, guaranteeing that $M^*$ is one of the MWPMs in $\overline{G}$ with weights $w$.

By the isolation lemma~\cite{mulmuley1987matching,Mulmuley1987}, we also have
\begin{align}
    \Pr[\text{Isolation fails to hold}] \le \frac{|\overline{V}|}{W_{\max}},
\end{align}
implying that, with high probability, the MWPM in $\overline{G}$ with weights $\tilde{w}$ is unique. 
From here onward, we assume $\tilde{w}(e)$ as the working weights, and the MWPM is unique.

\smallskip
\textit{Step~2.}
We construct the matrix $B$ as defined in~\eqref{Eq:Tuttemat2}.

\smallskip
\textit{Step~3.}
For this step of computing $\det(B)$ to determine the weight of the MWPM, we will show that~\eqref{Eq:minweight2} indeed yields the weight. 
Recall that
\begin{align}
    \det(B) = \sum_{\sigma} \operatorname{sign}(\sigma) \prod_{i=1}^{|\overline{V}|} b_{i,\sigma(i)},
\end{align}
where $\sigma$ is a permutation, $\operatorname{sign}(\sigma)$ is its sign, and $b_{i,j}$ denotes the $(i,j)$-th entry of $B$. Since we are working on $\mathbb{F}_2[X]/(X^{w_{\mathrm{th}}})$, we have
\begin{align}
    \det(B) = \sum_{\sigma} \prod_{i=1}^{|\overline{V}|} b_{i,\sigma(i)}.
\end{align}
Each $\sigma$ corresponds to a subgraph of $\overline{G}$, since $\{i,\sigma(i)\}$ represents an edge and every permutation contains at least one cycle. 
We partition the set of permutations into $S_2$, the subset consisting entirely of 2-cycles, and $S_{>2}$, the subset containing permutations with cycles of length greater than 2.

Then we have
\begin{align}
    \det(B) &= \sum_{\sigma \in S_2} \prod_{i=1}^{|\overline{V}|} b_{i,\sigma(i)} \nonumber  + \sum_{\sigma \in S_{>2}} \prod_{i=1}^{|\overline{V}|} b_{i,\sigma(i)}.
\end{align}

For $\sigma_{>2} \in S_{>2}$, there exists another permutation $\sigma_{>2}'$ obtained by reversing a cycle.
Since $b_{i,j} = b_{j,i}$, we have
\begin{align}
    \prod_{i=1}^{|\overline{V}|} b_{i,\sigma_{>2}(i)} = \prod_{i=1}^{|\overline{V}|} b_{i,\sigma_{>2}'(i)}.
\end{align}
Recalling that $1 + 1 = 0$ over $\mathbb{F}_2$, we see that the contributions from $S{>2}$ cancel out
\begin{align}
    \sum_{\sigma \in S_{>2}} \operatorname{sign}(\sigma) \prod_{i=1}^{|\overline{V}|} b_{i,\sigma(i)} = 0.
\end{align}

Computing the rest, we have
\begin{align}
    \sum_{\sigma \in S_2} \operatorname{sign}(\sigma)\prod_{i = 1}^{|\overline{V}|} b_{i, \sigma(i)}
    &= \sum_{\sigma \in S_2} X^{\sum_{i = 1}^{|\overline{V}|} \tilde{w}_{i, \sigma(i)}} \\
    &= \sum_{\sigma \in S_2} X^{2w^* + \epsilon(\sigma)}.
\end{align}
Here, $w^*$ denotes the weight of the MWPM in $\overline{G}$ with weights $\tilde{w}$, and $\epsilon(\sigma)$ over $\sigma \in S_2$ is defined as
\begin{align}
    \epsilon(\sigma) \coloneqq \sum_{i = 1}^{|\overline{V}|} \tilde{w}_{i, \sigma(i)} - 2w^*.
\end{align}
Considering the graph corresponding to each permutation in $S_2$, from the properties of the MWPM, we have $\epsilon(\sigma) \in \mathbb{Z}$ and
\begin{align}
    \epsilon(\sigma) =
        \begin{cases}
        0 & (\sigma = \sigma_M),\\
        > 0 & (\sigma \neq \sigma_M),
        \end{cases}
\end{align}
where $\sigma_M$ denotes the permutation corresponding to the MWPM. This holds because a permutation composed of 2-cycles corresponds to a perfect matching in the graph.
Then we have
\begin{align}
    \det(B) = X^{2w^*}\sum_{\sigma \in S_2} X^{\epsilon(\sigma)}.
\end{align}
Under the assumption that the MWPM is unique, there exists exactly one permutation $\sigma$ such that $\epsilon(\sigma)=0$. Moreover, because $w_{\mathrm{th}} > 2w^*$, the term $X^{2w^*}$ remains in the truncated polynomial ring $\mathbb{F}_2[X]/(X^{w_{\mathrm{th}}})$. Therefore, $w^*$ indeed corresponds to the MWPM weight.

\smallskip
\textit{Step~4.}
We compute the minors for each edge and use the result from Step~3 to identify edges belonging to the MWPM. 
From the Laplace expansion
\begin{align}
    \left|M^B_{i,j}\right| \cdot X^{\tilde{w}(\{i,j\})} = \sum_{\sigma:\,\sigma(i)=j} 
    \operatorname{sign}(\sigma) \prod_{i=1}^{|\overline{V}|} b_{i,\sigma(i)}, 
\end{align}
we have
\begin{align}
    M^B_{i,j} \cdot X^{\tilde{w}(\{i,j\})} = \sum_{\sigma:\,\sigma(i)=j} \prod_{i=1}^{|\overline{V}|} b_{i,\sigma(i)}.
\end{align}
The contribution from permutations with cycles of length greater than 2 cancels out with that of its reverse, as argued above. 
Thus, we have
\begin{align}
    M^B_{i,j} \cdot X^{\tilde{w}(\{i,j\})}
    = X^{2w^*}\sum_{\sigma \in S_2:\, \sigma(i) = j} X^{\epsilon(\sigma)},
\end{align}
where $S_2$ denotes the set of permutations consisting of 2-cycles.  
Recalling the definition of $\epsilon(\sigma)$, for $\sigma$ that satisfy $\sigma(i) = j$, we have 
\begin{align}
    \epsilon(\sigma) =
    \begin{cases}
        0 & (\sigma = \sigma_M),\\
        > 0 & (\sigma \neq \sigma_M),
    \end{cases}
\end{align}
where $\sigma_M$ denotes the permutation corresponding to the MWPM. This implies that $\epsilon(\sigma)$ can be 0 if and only if $\sigma_M(i) = j$. Since the MWPM is unique, focusing on the lowest-degree term, we correctly identify the MWPM edges according to \eqref{Eq:setofMWPM2}. This confirms that our algorithm yields the same MWPM as the algorithm in Ref.~\cite{takada2025doubly} as long as $w_\mathrm{th}>2w^\ast$.

Hence, if $w_{\mathrm{th}} > 2w^*$, our algorithm outputs a unique and correct MWPM with high probability.
\end{proof}

We also show that if the overflow occurs, our algorithm can naturally detect it, as shown below.

\begin{theorem}[\label{thm:failure}Failure Detection]
    Let $w^*$ denote the weight of the MWPM for the input graph $\overline{G} = (\overline{V}, \overline{E})$ with edge weights $\tilde{w}: \overline{E} \to \mathbb{Z}$. If
    \begin{align}
        w^* \ge \frac{w_{\mathrm{th}}}{2},
    \end{align}
    the algorithm outputs a failure indicator.
\end{theorem}

\begin{proof}
From Theorem~\ref{thm:correctness}, we have
\begin{align}
    \det(B) &= X^{2w^*} \sum_{\sigma \in S_2} X^{\epsilon(\sigma)}, \\
    M^B_{i,j} \cdot X^{\tilde{w}(\{i,j\})} &= X^{2w^*} \sum_{\sigma \in S_2:\, \sigma(i) = j} X^{\epsilon(\sigma)}.
\end{align}
Since $2w^* \ge w_{\mathrm{th}}$ and $\epsilon(\sigma) \ge 0$, all terms vanish, causing both expressions to become zero. In this situation, \eqref{Eq:minweight2} evaluates to zero, the algorithm correctly outputs a failure indicator, signaling that the MWPM weight exceeds the representable range.
\end{proof}

\subsection{Numerical estimation of required parameters}\label{Sec:EvaluationTheoretically}
Here, we numerically evaluate the proposed MWPM decoder from the perspective of scaling accuracy.
In particular, since the edge weights are discretized as in \eqref{Eq:discretized}, we investigate how much scaling precision is sufficient to ensure that the decoded matching remains identical to the one obtained using floating-point weights.
Although our decoding problem is formulated in terms of integer-weighted path graphs as in~\eqref{Eq:discretized}, our numerical estimates aim to determine the arithmetic bit length required for these integer weights to achieve error-suppression performance comparable to that of conventional MWPM decoding methods based on floating-point weights, such as the ones in Refs.~\cite{wu2023fusion,Higgott2025sparseblossom}.

The numerical simulation was conducted for the $[[n,\,k=1,\,d=\sqrt{n}]]$ two-dimensional rotated surface codes~\cite{bombin2007optimal} under the independent and identically distributed (IID) circuit-level depolarizing error model with a physical error rate of $p = 10^{-3}$.
After preparation of $\ket{0}$, a bit flip may occur with probability $p$.
After a single-qubit gate, Pauli errors $X$, $Y\propto XZ$, and $Z$ may occur with probability $p/3$, respectively, where identity gates also suffer from errors.
After a two-qubit $\textsc{CNOT}$ gate, $X\otimes I,I\otimes X,X\otimes X,\ldots,Z\otimes Z$ may occur with probability $p/15$, respectively, where $I$ is the identity matrix.
Before a measurement in the $Z$ basis, a bit flip may occur with probability $p$.
We simulated memory experiments consisting of $d$ rounds of syndrome extraction, starting from a codeword of logical $\ket{0}$ and ending with measurements in the $Z$ basis. 

Under this setting, we estimate the scale of the discretized weights in \eqref{Eq:discretized} required to reproduce the same MWPM as that obtained with floating-point weights. 
In the simulation, we used Stim~\cite{gidney2021stim} to generate the syndrome extraction circuits. The detector error model obtained from Stim was converted into a graph using PyMatching~\cite{Higgott2025sparseblossom}. To prepare a version of the graph with integer edge weights, we further converted it into a NetworkX~\cite{hagberg2008exploring} graph. For the detector graph on NetworkX, we multiplied an integer scaling factor $C$ to the edge weights, where $C$ is defined as the smallest integer such that the minimum weight obtained from \eqref{Eq:discretized} becomes a $b$-digit binary number. This ensures that all edge weights have at least $b$ significant binary digits. The resulting integer-weighted detector graph was then converted back into a PyMatching graph, producing both floating-point and integer-weight versions of the same graph.
We sampled errors using Stim and performed MWPM decoding with PyMatching for both the floating-point and integer-weight graphs to check whether they produced identical matchings. This process was repeated $10^7$ times for each code distance $d \in \{3, 5, \dots, 29\}$ and for each binary precision $b \in \{5, 6, \dots, 11\}$ to determine the minimum number of binary digits required to reproduce the same MWPM as that obtained from the floating-point graph.

\begin{figure}[tbp]
  \centering
  \includegraphics[width=0.96\linewidth]{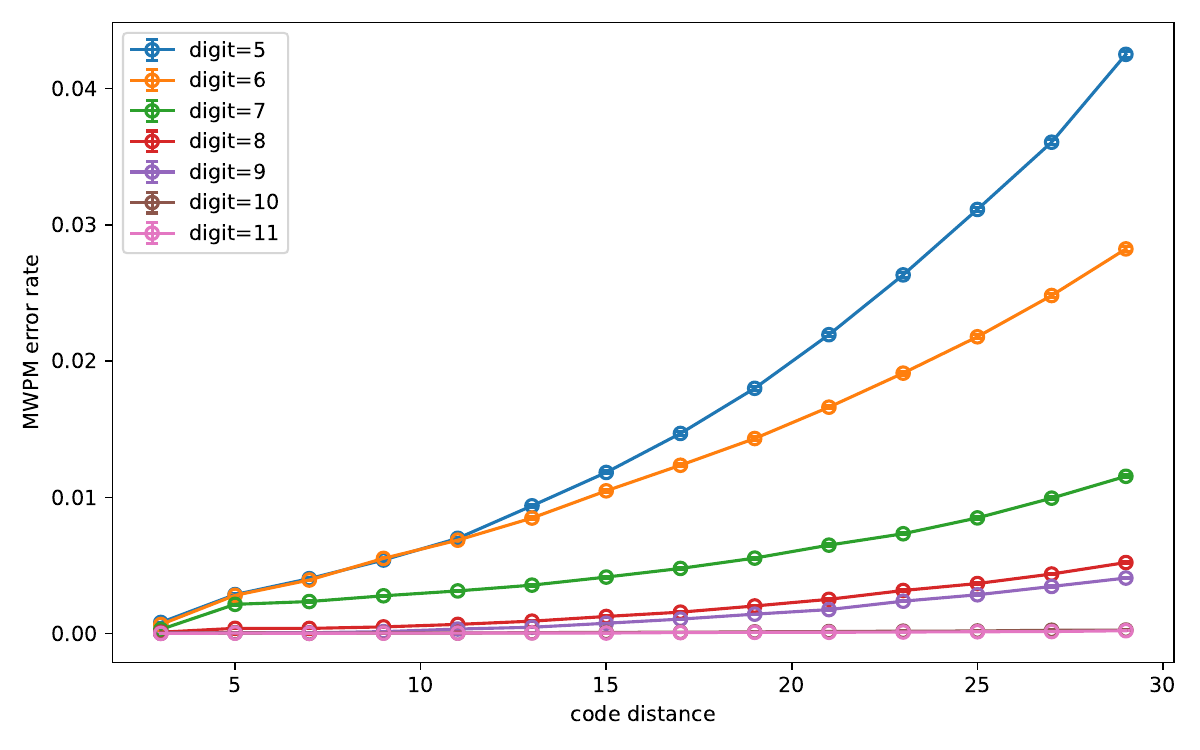}  \captionsetup{justification=raggedright,singlelinecheck=false}
  \caption{
    MWPM error rate as a function of the code distance for different discretization precisions of integerized edge weights. Each curve corresponds to a different number that represents a binary digit of the minimum discretized edge weight of a detector graph. The error rate indicates the fraction of trials where the MWPM solution obtained from the integer-weight graph differs from that obtained using the floating-point weights. Error bars indicate the standard error computed over the $10^7$ trials.
  }
  \label{fig:discretization_accuracy}
\end{figure}

In Fig.~\ref{fig:discretization_accuracy}, we plot the fraction of mismatched MWPMs between the floating-point and discretized-weight graphs as a function of the code distance for various discretization precisions. As the binary precision increases, the mismatch rate decreases rapidly, and for $b \geq 10$, the mismatch rate becomes zero, up to the statistical error below the logical error rate across small code distances $d = 5, 7, 9$. This result demonstrates that a precision of approximately 10 binary digits is sufficient to reproduce the same MWPM as the floating-point version under the considered noise model and parameter regime for small code distances. For sufficiently large code distances, the mismatch rate can be reduced below the logical error rate by increasing the binary precision.

The trend observed in Fig.~\ref{fig:discretization_accuracy} can be interpreted as follows. The MWPM algorithm depends only on the relative ordering of edge weights, not their absolute values. Therefore, discretization can change the MWPM result only when two or more competing matching paths have total weights that are very close. As the discretization precision increases, the relative ordering of these weights becomes stable, leading to the disappearance of mismatches.
Overall, these results indicate that MWPM decoding outcomes with floating-point weights can be faithfully reproduced by using only integer-weight graphs with a feasibly small binary precision for the scaled edge weights. 

\section{Heuristic modifications for practical bit-length reduction}\label{Sec:Heuristic_modification}

In this section, we propose heuristic modifications to our MWPM-decoding algorithm presented in Sec.~\ref{Sec:Theoretically_guaranteeed} and numerically evaluate the practical feasibility of its implementation.
In Sec.~\ref{Sec:HeuristicBitReduction}, we introduce heuristic modifications aimed at reducing the required arithmetic bit length, including a modified isolation strategy and a variable-precision scheme. Their performance is numerically evaluated in Sec.~\ref{Sec:EvaluationHeuristically}. In Sec.~\ref{Sec:ParameterChoice}, we propose practical parameter choices for hardware implementation. 

\subsection{Description of the heuristic improvements on the overflow-safe polylog-time MWPM decoding}\label{Sec:HeuristicBitReduction}

In this section, we present two strategies for reducing the required arithmetic bit length. In Sec.~\ref{Sec:IsolationBitReduction}, we remove the dominant source of bit-length blowup by modifying isolation. In Sec.~\ref{Sec:ScalingBitReduction}, we further introduce the variable-precision scheme that efficiently verifies the MWPM.

\subsubsection{Modification of isolation}\label{Sec:IsolationBitReduction}
The randomized algorithm described in Sec.~\ref{Sec:CorrectAlgorithm} isolates the MWPM by scaling every edge weight according to~\eqref{Eq:tildew}, where the amplification factor $\tilde{C}$ in~\eqref{Eq:tildeC} guarantees that, once the perturbed MWPM becomes unique, it coincides with the MWPM under the original weights $w : \overline{E} \to \mathbb{Z}$.
Although this construction ensures correctness with high probability, the large magnitude of $\tilde{C}$ substantially increases the bit length required to compute the determinant of matrix $B$ in \eqref{Eq:Tuttemat2}, creating a substantial bottleneck for practical implementations.

Here, we propose an alternative perturbation strategy by omitting the multiplication of the scaling factor $\tilde{C}$, which considerably reduces the required bit length. In particular, the algorithm proceeds identically to the procedure in Sec.~\ref{Sec:CorrectAlgorithm}, except for the perturbation rule and the selection of the final output.
We describe the modified step below.

\smallskip
\textit{Modified step.}
Instead of constructing the amplified weights $\tilde{w}$ in~\eqref{Eq:tildew}, we define perturbed weights without amplification by
\begin{align}
    \hat{w}(e) \coloneqq w(e) + W(e), \label{Eq:hatw}
\end{align}
where $W(e)$ is drawn in the same manner as in Sec.~\ref{Sec:CorrectAlgorithm}. All subsequent steps---construction of the matrix $B$ in \ref{Eq:Tuttemat2}, computation of $\det(B)$, extraction of minimum-degree terms, and evaluation of minors to obtain the set representing the unique 
MWPM---are performed exactly as in Sec.~\ref{Sec:CorrectAlgorithm}, with the replacement of $\tilde{w}(e)$ in~\eqref{Eq:tildew} with $\hat{w}(e)$ in~\eqref{Eq:hatw}
throughout.
In addition to this, for each matching obtained from $\hat{w}(e)$, the true weight of the matching with respect to the original, unperturbed weight function $w$ is computed.
When the perturbation procedure is repeated multiple times, we select as the final output the set $M$ corresponding to the perturbation whose resulting matching attains the minimum weight with respect to 
the original weights $w$.

\subsubsection{Modification of accuracy scaling}\label{Sec:ScalingBitReduction}
As described in Sec.~\ref{Sec:IsolationBitReduction}, the heuristic method outputs an estimate of the optimal solution by recomputing the weight of each MWPM candidate using the original weights.
In this process, graphs of the same precision are used both to generate MWPM candidates and to evaluate the final weights. As a result, while the probability of obtaining candidates that include the correct MWPM increases, the bit length required for determinant computations increases to maintain high precision, which may still make the implementation impractical.

Here, we further extend this method by allowing the discretization precisions in~\eqref{Eq:discretized} to vary and adapting them at each step.
The algorithm proceeds identically to the procedure in Sec.~\ref{Sec:IsolationBitReduction}, except for the weight functions used in the computation. We describe the modified step below.

\smallskip
\paragraph*{Modified step.}
Let $\overline{G}=\qty(\overline{V}, \overline{E})$ be the path graph. First, we define the discretized weight as
\begin{equation}
    w_C(e) \coloneqq \left\lceil -C \ln\!\left(\sum_j p_j\right) \right\rceil,
\end{equation}
by explicitly introducing the scaling factor $C$, where the notation follows that of \eqref{Eq:discretized}. We introduce two scaling factors that satisfy $C_{\mathrm{high}} > C_{\mathrm{low}}$ and prepare the corresponding weight functions
\begin{gather}
    w_{C_{\mathrm{high}}} : \overline{E} \rightarrow \mathbb{Z}, \label{Eq:w_high}\\
    w_{C_{\mathrm{low}}}  : \overline{E}  \rightarrow \mathbb{Z} \label{Eq:w_low}.
\end{gather}
With these definitions, the proposed method can be described as the following replacement of weight functions in the existing algorithm. In the step of generating MWPM candidates, namely in the computation of the MWPM in Sec.~\ref{Sec:CorrectAlgorithm}, the weight function $w_{C_{\mathrm{low}}}$ is used. On the other hand, in the step of selecting the final output, that is, in the additional weight evaluation introduced in Sec.~\ref{Sec:IsolationBitReduction}, the weight function $w_{C_{\mathrm{high}}}$ is used.

\subsection{Numerical evaluation of the heuristic decoder
}\label{Sec:EvaluationHeuristically}
Here, we evaluate the proposed improvements on our MWPM decoder from the following perspectives.
\begin{enumerate}
    \item \textbf{Required bit length:}
    In Sec.~\ref{Sec:CorrectAlgorithm} and Sec.~\ref{Sec:HeuristicBitReduction}, we presented two approaches: a theoretically correct formulation and its heuristic, yet practically optimized variant with reduced bit length.
    In Sec.~\ref{Sec:EvaluationRequiredBitLength}, we estimate the required arithmetic bit length $w_{\mathrm{th}}$ to reliably obtain the correct MWPM without affecting the logical error rates. 

    \item \textbf{Correctness of the MWPM:}
    In the method described in Sec.~\ref{Sec:HeuristicBitReduction}, MWPM candidates are obtained via heuristic approximations. In Sec.~\ref{Sec:EvaluationCorrectness}, we numerically evaluate how $w_{\mathrm{th}}$ affects the correctness of the MWPM and how the number of perturbation trials influences the probability of obtaining the correct MWPM\@.
\end{enumerate}
All the numerical simulations were conducted in the same setting as in Sec.~\ref{Sec:EvaluationTheoretically}\@.

\subsubsection{Required bit length}\label{Sec:EvaluationRequiredBitLength}
We estimate the threshold bit length $w_{\mathrm{th}}$ in Sec.~\ref{Sec:CorrectAlgorithm} required to obtain the correct MWPM with high probability.

In this simulation, we used the same procedure as in Sec.~\ref{Sec:EvaluationTheoretically} to generate circuits with Stim, construct the detector graph using PyMatching, and produce its integer-weight version for MWPM decoding using NetworkX\@. For efficiency, we fixed the code distance to $d = 5$. Although larger code distances require a digit length of at least $10$ for sufficient numerical precision, a binary precision of $b = 8$ is sufficient for $d = 5$, and we therefore used $b = 8$ as a high binary precision. To evaluate the method in Sec.~\ref{Sec:ScalingBitReduction}, we also adopted a reduced binary precision, or low binary precision, of $b = 4$ in order to further suppress the required bit length to a practically feasible level.
Under this setting, we sampled errors using Stim and performed MWPM decoding with PyMatching on the graphs to obtain the total weight of the resulting MWPM\@. From this MWPM weight, we applied the two scaling procedures corresponding to \eqref{Eq:tildew} and \eqref{Eq:hatw}, respectively, thereby generating the two types of scaled weights, $\tilde{w}_{\mathrm{total}}$ and $\hat{w}_{\mathrm{total}}$.
In summary, we compare three cases: the original scaling $\tilde{w}$, the proposed scaling $\hat{w}$ with high binary precision, and the same scaling $\hat{w}$ with low binary precision.

To ensure that the perturbed weights lead to a correct MWPM, we adopted the perturbation range 
\begin{align}
    1 \le W(e) \le W_{\max} \qty(= \left\lceil 0.8\, n^{0.8} \right\rceil),
\end{align}
where $n$ denotes the size of a path graph. We chose a looser bound than that in Ref.~\cite{takada2025doubly} to ensure a sufficient margin. The effect of the perturbation was incorporated by using the maximum value over this allowed range. We then computed $2\tilde{w}_{\mathrm{total}} + 1$ (or the corresponding quantity for $\hat{w}_{\mathrm{total}}$) for each trial. From Theorem~\ref{thm:correctness}, this value provides a lower bound on the required threshold bit length $w_{\mathrm{th}}$. This process was repeated $10^7$ times to determine and compare the lower bounds of the required threshold bit length $w_{\mathrm{th}}$ for the two algorithms described in Secs.~\ref{Sec:CorrectAlgorithm} and~\ref{Sec:HeuristicBitReduction}.

\begin{figure}[tbp]
  \centering
  \includegraphics[width=0.96\linewidth]{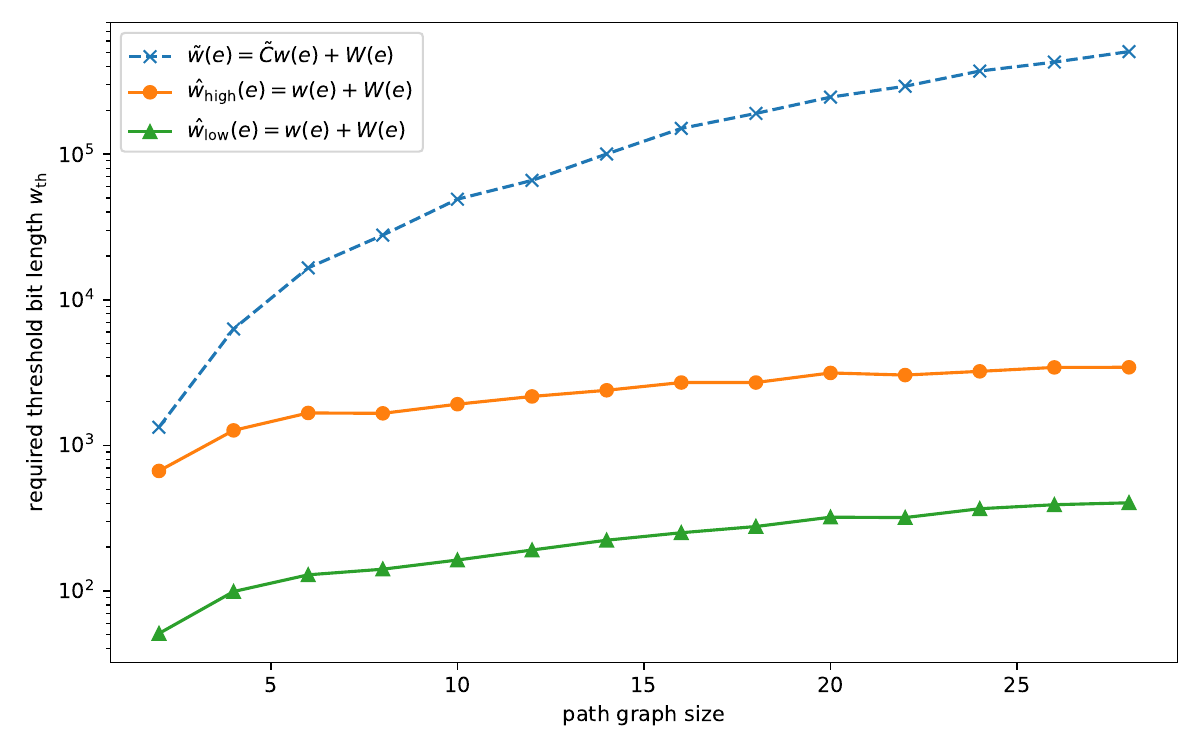}  \captionsetup{justification=raggedright,singlelinecheck=false}
  \caption{
    Required $w_{\mathrm{th}}$ to obtain the correct MWPM for each number $\{2,4,\dots,28\}$ with a code distance of 5, and a high binary precision of 8 bits and a low binary precision of 4 bits. The three curves represent the lower bounds of $w_\mathrm{th}$ obtained from MWPM weights using the respective edge-weight scalings, where the plotted values correspond to the maximum $w_\mathrm{th}$ observed over $10^7$ trials
  }
  \label{fig:path_graph_size_vs_weight}
\end{figure}

In Fig.~\ref{fig:path_graph_size_vs_weight}, we plot lower bounds of the required threshold bit length $w_{\mathrm{th}}$, obtained from the maximum $w_{\mathrm{th}}$ that we found within $10^7$ trials for each path graph size. Note that, by construction, the path graph size is restricted to even integers, as illustrated in Ref.~\cite{takada2025doubly}. As the path graph size increases, the gap between the curves corresponding to $\tilde{w}$ and $\hat{w}$ becomes larger. For instance, at a path graph size of $n = 28$, the scaling method based on \eqref{Eq:tildew} requires a threshold bit length of about $6 \times 10^5$, whereas the method described in \eqref{Eq:hatw} requires only about $4 \times 10^{3}$ when using high binary precision. Furthermore, using the low binary precision $b = 4$, we achieve further reduction in the required threshold bit length, down to approximately $5 \times 10^2$,
corresponding to a reduction of more than $99.9\%$ relative to $6\times 10^5$.

These results indicate that the method described in Sec.~\ref{Sec:CorrectAlgorithm} requires a bit length that may still require impractically large for FPGA-based implementations, whereas the bit-reduction techniques introduced in Secs.~\ref{Sec:IsolationBitReduction} and~\ref{Sec:ScalingBitReduction} bring the required threshold bit length within a feasible range for FPGA hardware that supports $512$-bit computing. 

\subsubsection{Correctness of the MWPM} \label{Sec:EvaluationCorrectness}
We estimate the probability that the algorithm presented in Sec.~\ref{Sec:HeuristicBitReduction} does not output the correct MWPM\@.

In this simulation, we used the same procedure as in Sec.~\ref{Sec:EvaluationRequiredBitLength} to generate circuits with Stim, construct the detector graph using PyMatching, and produce its integer-weight version for MWPM decoding using NetworkX, where we fixed the code distance to $d = 5$, chose a high binary precision of $b = 8$, and a low binary precision of $b = 4$. We performed a Monte Carlo simulation to evaluate the probability that our algorithm fails to output an MWPM of the same weight as that of the conventional blossom algorithm~\cite{Higgott2025sparseblossom}.

In this simulation, we first constructed a lookup table from the detector graph using NetworkX\@. We then sampled detection events with Stim and constructed the corresponding path graph by reading the lookup table. The Mersenne Twister~\cite{matsumoto1998mersenne} was used as the pseudo-random function to generate random numbers for perturbation. The seed values are chosen uniformly at random from the set of $32$-bit unsigned integers. For each perturbation set, we constructed the weight-perturbed path graph following Sec.~\ref{Sec:IsolationBitReduction} and computed the determinant. For this computation, we used an implementation based on Ref.~\cite{takada2025doubly}, modified to retain only the first $w_{\mathrm{th}}$ digits during arithmetic to evaluate the success probability as a function of $w_{\mathrm{th}}$. For every perturbation set, we checked whether the value $w^*$ obtained from \eqref{Eq:minweight2} coincides with the weight computed from \eqref{Eq:setofMWPM2}; those satisfying this condition are regarded as MWPM candidates. Following the methods described in Sec.~\ref{Sec:HeuristicBitReduction}, we evaluated the correctness of the algorithm using the following two methods.

\paragraph{Base modified method: isolation-based bit reduction.}
For each MWPM candidate that is obtained from the high-binary-precision path graph, its weight on the original high-binary-precision path graph is recomputed, and the one with the minimum weight is selected.

\paragraph{Extended modified method: low-precision candidate generation with high-precision verification.}
To further reduce the required bit length, we additionally consider a variant in which MWPM candidates are generated using the low-binary-precision path graph. For each MWPM candidate obtained in this way, its weight on the original high-binary-precision path graph is recomputed, and the one with the minimum weight is selected.

We then verified whether its weight coincides with the MWPM weight obtained by PyMatching, i.e., whether it corresponds to the correct MWPM on the integer-scaled path graph.

These entire procedures were repeated $10^6$ times for each chosen number of perturbation samples, allowing us to estimate how many perturbation samples are sufficient for our bit-reduction methods. Since the failure probability only needs to be kept below the logical error rate, we decoded path graphs of size at most $28$. Following Ref.~\cite{takada2025doubly}, we adopted $W_{\max}$, which denotes the upper bound of the perturbation range, as the number of perturbation sets, which also indicates the number of weight-perturbed path graphs. As in Sec.~\ref{Sec:EvaluationRequiredBitLength}, we used $\lceil 0.8\, n^{0.8} \rceil$ to ensure that the correct MWPM can be obtained. Using this value as the base perturbation, we conducted additional experiments by doubling the number of perturbation sets while keeping the perturbation range fixed. Note that, throughout our experiments, we did not directly check whether the MWPM is uniquely isolated or not, since the isolation is only a sufficient condition, not a necessary condition, for finding the correct MWPM\@.

\begin{figure*}[tbp]
  \centering
  \begin{subfigure}{3.4in}
    \centering
    \includegraphics[width=3.4in]{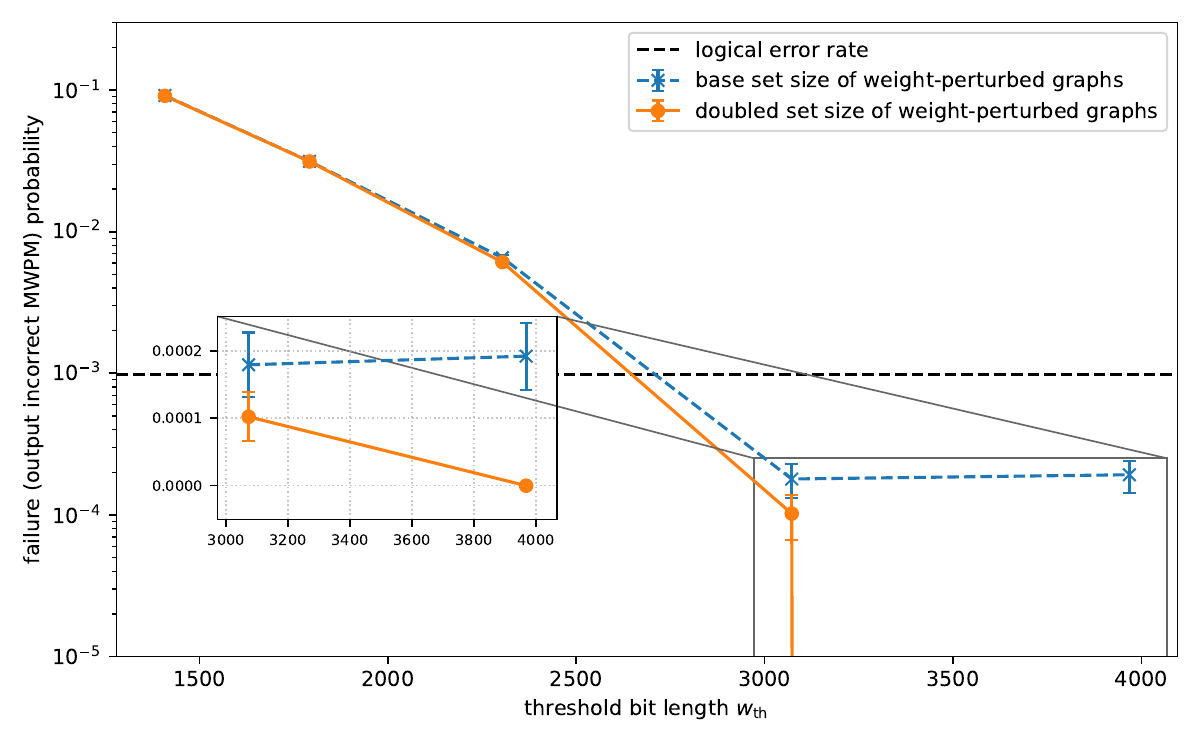}
    \caption{Base modification}
    \label{fig:threshold_vs_failure_rate_a}
  \end{subfigure}
  \hfill
  \begin{subfigure}{3.4in}
    \centering
    \includegraphics[width=3.4in]{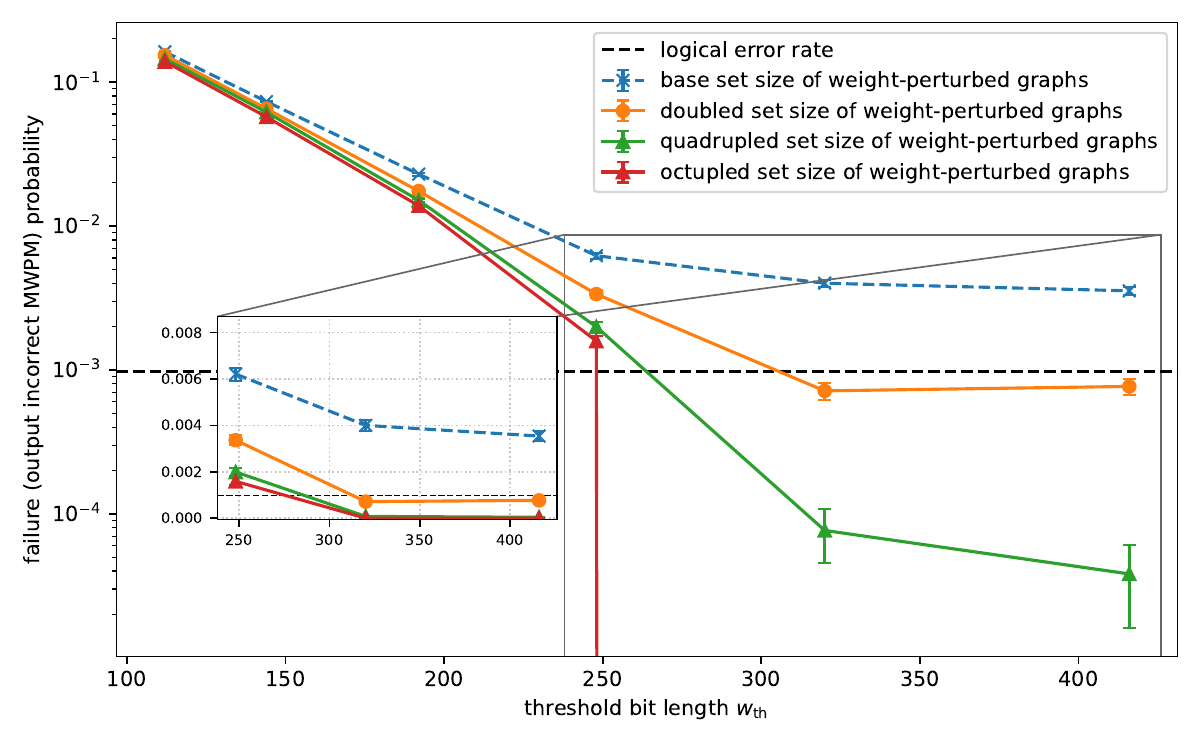}
    \caption{Extended modification}
    \label{fig:threshold_vs_failure_rate_b}
  \end{subfigure}
  \captionsetup{justification=raggedright,singlelinecheck=false}
  \caption{    
    Failure probability as a function of the threshold bit length $w_{\mathrm{th}}$ for a code distance of 5, a high binary precision of 8 bits, and a low binary precision of 4 bits. The failure probability represents the fraction of trials, among those with path graph size at most 28, in which the decoding process fails, or the decoded MWPM of the perturbed path graph does not coincide with any of the correct MWPMs of the original path graph. Error bars indicate the standard error computed over $10^6$ trials. 
    (a) Isolation-based bit-reduction method using high-binary-precision MWPM candidate generation. 
    (b) Extended method using low-binary-precision MWPM candidate generation with high-binary-precision verification. 
    For each panel, the dashed horizontal line indicates the logical error rate. The blue curve denotes the baseline number of perturbation samples required by the conventional method, while orange, green, and red denote 2×, 4×, and 8× that number, respectively.
    Our method sufficiently suppresses the decoding failure probability with arithmetic bit lengths of $3\times 10^2$–$5\times 10^2$, as indicated by the green curve.
    }
  \label{fig:threshold_vs_failure_rate}
\end{figure*}

In Fig.~\ref{fig:threshold_vs_failure_rate}, we plot the failure probability for each threshold bit length $w_{\mathrm{th}}$. Figure~\ref{fig:threshold_vs_failure_rate_a} shows the results for the base isolation-based bit-reduction method, while Fig.~\ref{fig:threshold_vs_failure_rate_b} shows the results for the extended method using low-precision MWPM candidate generation. In both figures, the gap between the curves is zero up to the error bar for small values of $w_{\mathrm{th}}$.
As $w_{\mathrm{th}}$ increases, the gap gradually becomes more pronounced.
In Fig.~\ref{fig:threshold_vs_failure_rate_a}, we observe that the failure probability for the base value remains almost constant below the logical error rate in the range $3000 \lesssim w_{\mathrm{th}} \lesssim 4000$, while for the doubled value, the failure probability drops to zero at $w_{\mathrm{th}} \approx 4000$. In Fig.~\ref{fig:threshold_vs_failure_rate_b}, we observe a trade-off between the failure probability and the set size: at the doubled set size, the failure probability remains below the logical error rate, and when the set size is octupled, it drops to zero.

Based on the results presented in Sec.~\ref{Sec:EvaluationRequiredBitLength} and the numerical simulation above, the base isolation-based bit-reduction method requires a threshold bit length of $w_{\mathrm{th}} \approx 4 \times 10^3$, whereas the original method requires $w_{\mathrm{th}} \approx 6 \times 10^5$.
This corresponds to an estimated improvement in bit length of about $99.3 \%$ compared with the previous work~\cite{takada2025doubly}.
By additionally using the extended method with low-precision MWPM candidate generation, the required threshold bit length is further reduced to approximately $3\times 10^2$–$5\times 10^2$. This range is consistent with the numerics in Sec.~\ref {Sec:EvaluationRequiredBitLength}, and the numerics here confirm that it keeps the failure probability below the logical error rate.
As in Sec.~\ref {Sec:EvaluationRequiredBitLength}, this corresponds to an improvement of about $99.9\%$ compared with the original method in Ref.~\cite{takada2025doubly}.

\subsection{Proposal of parameter choice for implementation}\label{Sec:ParameterChoice}
Based on the theoretical analysis and numerical simulations presented in the previous sections, we summarize practical guidelines for parameter selection when implementing the proposed decoder on hardware, with particular emphasis on proof-of-principle demonstrations at code distance $d=5$:

\begin{itemize}
    \item arithmetic bit length: $w_{\mathrm{th}} \ge 5 \times 10^2$ bits;
    \item low binary precision: $4$ bits;
    \item high binary precision: $\ge 8$ bits;
    \item range of perturbations: $\ge \left\lceil 0.8\, n^{0.8} \right\rceil$;
    \item the number of perturbation sets: $\ge 8W_{\max}$.
\end{itemize}

\section{Discussion}\label{Sec:Discussion}
In this work, we address fundamental implementability challenges associated with the determinant-based algorithm for solving MWPM in polylog parallel time.
We have established a framework that guarantees the mathematical correctness of this approach under finite bit-length arithmetic.
Furthermore, we have mitigated the large bit-length requirements of the previous determinant-based algorithm in Ref.~\cite{takada2025doubly} by identifying and removing computational bottlenecks, thereby substantially reducing the arithmetic bit length required for its implementation.
These results close a practical gap between asymptotically fast MWPM decoding theory and its finite-precision hardware implementation.

Conventionally, MWPM decoders are regarded as having polynomial-time complexity, and achieving polylog runtime for MWPM decoding has been challenging even with parallelization. By contrast, the polylog-time parallel MWPM decoder proposed in Ref.\cite{takada2025doubly} offers a route toward doubly-polylog-time-overhead FTQC, substantially reducing the time overhead of FTQC\@. However, the approach in Ref.~\cite{takada2025doubly} has been hard to implement in practice because of an extremely large arithmetic bit length.
Our contribution is to provide concrete evidence that this asymptotically fast MWPM-decoding approach can be implemented at an experimental scale within realistic hardware constraints.
Analogous to recent surface-code experiments that validated exponential error suppression at increasing code distances~\cite{google2023suppressing,google2025quantum}, our results enable proof-of-principle demonstrations in the early-FTQC regime that can experimentally validate the possibility of asymptotically sublinear runtime scaling of the MWPM decoding.

Beyond the near-term regime, further optimization of the determinant-based MWPM-decoding methods and a thorough comparison with those based on variants of the blossom algorithm remain interesting directions for future work.
Indeed, variants of the blossom algorithm, such as the sparse blossom~\cite{Higgott2025sparseblossom} and the fusion blossom~\cite{wu2023fusion}, explicitly leverage the problem structure of surface-code decoding to reduce the computational complexity to approximately $O(n^{1.3})$, compared with naively invoking a general-purpose blossom-algorithm-based MWPM solver.
Similarly, the use of the problem structure may also reduce the computational complexity of the determinant-based approach.
Given the polylog runtime scaling of the determinant-based approach, which is asymptotically faster than blossom-based methods, it would be interesting to estimate the crossover point at which it becomes faster in practice, accounting for such algorithmic optimizations.

\begin{acknowledgements}
The authors acknowledge Teruo Tanimoto for discussion.
This work was supported by JST PRESTO Grant Number JPMJPR201A, JPMJPR23FC, JSPS KAKENHI Grant Number JP23K19970, JST CREST Grant Number JPMJCR25I5, and Faculty Research Funding from Google Quantum AI\@. 
\end{acknowledgements}

\clearpage
\appendix
\section{A parallel algorithm for finding an MWPM via conventional integer arithmetic}\label{appendix:A}
For completeness, we summarize the determinant-based polylog-time parallel algorithm in Ref.~\cite{takada2025doubly}.
In particular, given a path graph $\overline{G}=\qty(\overline{V}, \overline{E})$ with edge weights $w: \overline{E} \xrightarrow{} \mathbb{Z}$, we describe an algorithm for finding an MWPM in $\overline{G}$. 

\begin{enumerate}
  \item Adjust each edge weight and replace $w$ by $\tilde{w}$ for each $e \in \overline{E}$ as
  \begin{align}
      \tilde{w}(e) \coloneqq \tilde{C}w(e) + W(e), 
  \end{align}
  where we take the factor $\tilde{C}$ and the perturbation function $W$ as
  \begin{align}
      &\tilde{C} \coloneqq \frac{|\overline{V}|}{2}(W_{\max} - 1) + 1, \\
      &W: \overline{E} \xrightarrow{} \mathbb{N}, \\
      &1\leq W(e) \leq W_{\max}, \\
      &W_{\max} \coloneqq \max_{e \in \overline{E}}\{W(e)\}.
  \end{align}
  With this perturbation, with high probability, one of the MWPMs in $\overline{G}$ becomes a unique MWPM in the resulting weighted graph\@.
  \item Use a Tutte matrix $A$ with its element $A_{i, j}$ given by 
  \begin{align}
      A_{i, j} = 
      \begin{cases}
          x_{i, j} & i < j,\ \{i, j\} \in \overline{E}, \\
          -x_{i, j} & i > j,\ \{i, j\} \in \overline{E}, \\
          0 & \text{otherwise},
      \end{cases}
  \end{align}
  where $x_{i, j}$ is a variable. Then, for each edge $\{i, j\} \in \overline{E}$, we assign 
  \begin{align}
      x_{i, j} = 2^{\tilde{w}(\{i, j\})}, 
  \end{align}
  and obtain a matrix $B$ defined as
  \begin{align}
      B_{i, j} = 
      \begin{cases}
          2^{\tilde{w}(\{i, j\})} &i < j,\ \{i, j\} \in \overline{E}, \\
          -2^{\tilde{w}(\{i, j\})} &i > j,\ \{i, j\} \in \overline{E}, \\
          0 & \text{otherwise}, \\
      \end{cases}
      \label{Eq:Tuttemat1}
  \end{align}
  \item Compute $\det(B)$ and obtain
  \begin{align}
      w^* \coloneqq \max\left\{w \in \mathbb{Z}\middle| \frac{\det(B)}{2^{2w}} \in \mathbb{N}\right\}. \label{Eq:minweight1}
  \end{align}
  Note that the determinant can be computed in polylog parallel runtime by the Samuelson-Berkowitz algorithm~\cite{berkowitz1984computing} (see also Sec.~S2 of Ref.~\cite{takada2025doubly} for the details of this method). 
  \item For each edge, in parallel, calculate a minor $M^B_{i, j} = \det(B^{(i, j)})$, where $B^{(i, j)}$ is defined as an $\qty(|\overline{V}| - 1) \times \qty(|\overline{V}| - 1)$ matrix obtained from $B$ by removing the $i$th row and the $j$th column. Finally, we obtain
  \begin{align}
      M = \left\{\{i, j\}\middle|
      \frac{M^B_{i, j} \cdot 2^{\tilde{w}(\{i, j\})}}{2^{2w^*}} \text{ is odd}\right\} \subseteq \overline{E}. \label{Eq:setofMWPM1}
  \end{align}
\end{enumerate}

This subset $M$ represents the unique MWPM of $\overline{G}$.

\bibliography{citation} 

\end{document}